\documentclass[letterpaper,12pt]{amsart}

\usepackage{latexsym, amsmath, amssymb,amsthm,verbatim, xy}
\xyoption{all}

\usepackage{tgpagella}
\usepackage{mathpazo}

\usepackage[left=.95in,top=1in,right=.95in,bottom=1in]{geometry}
\usepackage{setspace,color}
\usepackage{graphicx}
\usepackage{todonotes}
\usepackage{enumerate}
\usepackage{natbib}
\usepackage[multiple]{footmisc}
\usepackage[leqno]{amsmath}
\usepackage{mathtools}

\usepackage{url}

\usepackage[breaklinks]{hyperref}

\usepackage{pgf,tikz}
\usetikzlibrary{arrows}
\definecolor{ffqqqq}{rgb}{1.,0.,0.}
\definecolor{ududff}{rgb}{0.30196078431372547,0.30196078431372547,1.}
\definecolor{burgundy}{rgb}{128, 0, 32}
\definecolor{dukeblue}{rgb}{0.0, 0.0, 0.61}
\definecolor{celestialblue}{rgb}{0.29, 0.59, 0.82}

\hypersetup{colorlinks,
            linkcolor=red,
            anchorcolor=red,
            citecolor=red,urlcolor={red}}

\doublespace

\usepackage{etoolbox}
\patchcmd{\section}{\scshape}{\bfseries}{}{}
\makeatletter
\renewcommand{\@secnumfont}{\bfseries}
\makeatother

\renewenvironment{quote}{%
   \list{}{%
     \leftmargin0.75cm   
     \rightmargin0.5cm
   }
   \item\relax
}
{\endlist}

\usepackage{tikz}

\newtheorem{theorem}{Theorem}

\newtheorem{proposition}{Proposition}

\theoremstyle{definition}
\newtheorem{remark}{Remark}
\newtheorem{example}{Example}

\newtheorem{definition}{Definition}

\renewenvironment{quote}{%
   \list{}{%
     \leftmargin0.75cm   
     \rightmargin0.5cm
   }
   \item\relax
}
{\endlist}

\begin{document}
\title{Complementarities in childcare allocation under priorities}
\author[Atay and Romero-Medina]{Ata Atay \and Antonio Romero-Medina}\thanks{Ata Atay is a Serra H\'{u}nter Fellow (Professor Lector Serra H\'{u}nter). Ata Atay gratefully acknowledges the financial support by the Spanish Ministerio de Ciencia e Innovaci\'{o}n through grant PID2020-113110GB-100/AEI/10.13039/501100011033, by the Generalitat de Catalunya through grant 2021SGR00306, and by the University of Barcelona through grant AS017672. Antonio Romero-Medina acknowledges the financial support by  Spanish Ministerio de Ciencia e Innovaci\'{o}n through grant AEI PID2020-118022GB-I0/AEI/10.13039/501100011033, and Comunidad de Madrid, Grants EPUC3M11 (V PRICIT) and H2019/HUM-589. We thank José Alcalde, Antonio Miralles, Jordi Mass\'{o}, and Matteo Triossi.\\ 
{\bf{Date: \today}}\\
\textbf{Atay:} Department of Mathematical Economics, Finance and Actuarial Sciences, and Barcelona Economic Analysis Team (BEAT), University of Barcelona, Spain. E-mail: \href{mailto:aatay@ub.edu}{\texttt{aatay@ub.edu}}. \\ \textbf{Romero-Medina:} Department of Economics, UC3M, Spain. E-mail: \href{mailto:aromero@eco.uc3m.es}{\texttt{aromero@eco.uc3m.es}}}

\begin{abstract}
We investigate the allocation of children to childcare facilities and propose solutions to overcome limitations in the current allocation mechanism. We introduce a natural preference domain and a priority structure that address these setbacks, aiming to enhance the allocation process. To achieve this, we present an adaptation of the Deferred Acceptance mechanism to our problem, which ensures strategy-proofness within our preference domain and yields the student-optimal stable matching. Finally, we provide a maximal domain for the existence of stable matchings using the properties that define our natural preference domain. Our results have practical implications for allocating indivisible bundles with complementarities.
\end{abstract}
\maketitle


\medskip

\hfill \emph{"Childcare is a textbook example of a broken market."}\footnote{Remarks by Secretary of the Treasury Janet L. Yellen on Shortages in the Child Care System (The U.S. Department of Treasury, 2021) reported in The Atlantic "The Case for Public Child Care" on January 5, 2023, \href{https://www.theatlantic.com/family/archive/2023/01/america-public-child-care-programs/672637/}{https://www.theatlantic.com/family/archive/2023/01/america-public-child-care-programs/672637/}. Accessed on: February 6, 2023.\nocite{dot21}}

\section{Introduction}

Childcare facilities provide parents an alternative to home care and facilitate children's cognitive skills development. In many countries, local governments offer and manage childcare services publicly. Still, these centralized systems are often met with dissatisfaction and can be a contentious political issue (see, for instance, \citealp{kk20}). Our motivation is the childcare allocation of the Community of Madrid (one of the 17 autonomous communities in Spain), which reflects the main features of childcare allocation procedures in Spain. 

Childcare is not mandatory in Spain, and the scarcity of available slots is well-documented. For instance, a survey conducted by the National Institute of Statistics in 2010 revealed that many mothers were compelled to quit their jobs or reduce their working hours due to inadequate childcare options. In the Community of Madrid, the enrollment rate in the academic year 2020-2021 was 18.3 \% for children under one year, 49.9 \% for children under two years, and 66.6 \% for children under three years. \footnote{Source: Datos y Cifras de la Educación 2020-2021. \href{http://www.madrid.org/bvirtual/BVCM050236.pdf}{http://www.madrid.org/bvirtual/BVCM050236.pdf\nocite{madrid21}}.	Accessed on May 27, 2022.}
The situation is not unique to Madrid. In France, public childcare slots frequently experience shortages, particularly in public \textit{crèche} facilities catering to children under three years of age  (\citealp{oecd}).

Parents tend to prefer enrolling their children in childcare at a later stage  (\citealp{mbg18, mle20}). However, they have incentives to secure a childcare slot as early as possible.\footnote{ Evidence of this behavior is frequent; we can find it in conversations with parents and the press. For example, \emph{``Only 11\% of applicants for slots in municipal schools have been admitted. More than 1000 children have been left without a slot at the schools in the Community of Madrid.''} \href{http://fuencarralelpardo.com/2021/09/10/3-600-ninos-sin-plaza-en-las-escuelas-infantiles-publicas-de-fuencarral-el-pardo/\%20}{http://fuencarralelpardo.com/2021/09/10/3-600-ninos-sin-plaza-en-las-escuelas-infantiles-publicas-de-fuencarral-el-pardo/\nocite{fuen21}}. Accessed on December 11, 2021.} There is a consensus among parents that: \emph{``If you want your child to enter a public childcare that you like, try to sign him/her up before (s)he is born.''} \footnote{\href{https://www.elconfidencial.com/espana/2017-05-28/educacion-temprana-acceso-guarderia-publica-lista-espera_1389553/}{https://www.elconfidencial.com/espana/2017-05-28/educacion-temprana-acceso-guarderia-publica-lista-espera\_1389553/\nocite{elcon}}. Accessed on December 11, 2021.} This is because securing a childcare slot grants absolute priority for enrollment in subsequent years, and the earlier parents apply, the higher their chances of obtaining one.

When students gain admission to the first year at a childcare facility, they acquire two objects: a slot for the current year and the highest priority claim for a slot in the following year. Typically, the preference for these two objects exhibits a strong positive correlation. In such cases, the allocation process remains unaffected by distortions. However, parents often demand more childcare slots than necessary for strategic reasons. Even if they do not intend to enroll their child in the first year, participating in the first-year allocation increases the likelihood of securing a slot for the subsequent year. The presence of a limited number of strategic applicants, combined with resource scarcity, can result in significant efficiency losses. 

This paper examines the distortion caused by history-dependent priorities in the context of childcare allocation. Following the minimalist market design approach of \cite{s23}, we aim to identify the aspects the current institution design fails to satisfy and provide alternatives that better fulfill the key objectives.  We adopt a two-period allocation process\footnote{We focus on a two-period allocation for simplicity, but our results can be extended to an arbitrary number of periods, as students retain priority at a school as long as they remain enrolled.} in which slots for the first period and priorities for the second period are allocated separately. We demonstrate that eliminating inter-temporal linkage can prevent efficiency loss as agents cannot manipulate their preferences to enhance their chances of securing a second-period allocation.

Following  \cite{kk05} and \cite{kkn09}, we expand individuals' preferences to include choices regarding childcare facilities during this period and priorities for the subsequent year. We assume that preferences for schools remain constant across both periods. Overcoming the challenge of allowing agents to exhibit strong intertemporal complementarities in their preferences over school slots while breaking the intertemporal allocation linkage, we enable agents to report strong complementarities between this year's school slots and next year's priority within the same childcare facility. This property is referred to as the willingness-to-remain property.

We adapt the student-proposing deferred acceptance ($SPDA$) of \cite{gs62} to our problem. We show that the adapted $SPDA$ ($aSPDA$) mechanism not only returns a stable matching within our preference domain (Theorem \ref{thm:stable}) but also remains strategy-proof for students (Theorem \ref{thm:sp}). Finally, we show that our preference domain is a maximal domain for the existence of stable matchings (Proposition \ref{pro:max_dom}).

Our findings contribute to the literature on childcare allocation problems in three significant ways: Firstly, our model differentiates between allocation and priority as distinct entities, capturing the dynamic nature of the agents' decision problem. Secondly, students are assigned in cohorts, and our expanded preference domain accommodates complementarities. Thus, our model is not encompassed by the many-to-many or matching with couples frameworks. Thirdly, we propose a stable and strategy-proof mechanism within the extended preferences domain and priority structure.
\subsection{Related literature}

The childcare allocation problem is dynamic, as children admitted to a childcare facility retain their slots until they start kindergarten. Unlike other countries, in Spain, childcare admission follows a cohort system similar to the school admission problem (see \citealp{as03}). Students within a cohort interact in the successive admission process. However, in the childcare admission process, unlike the school admission problem, parents determine when each child starts attending childcare. The closest paper to ours is \cite{kmt14}, which examines the allocation of children to public childcare facilities in Denmark. In their setting, admissions occur monthly, and children's priority depends on their age.

In contrast, our setting involves yearly cohorts, with no competition among children of different cohorts. While \cite{kmt19} studies the strategic behavior in a dynamic deferred acceptance (DA) mechanism under priority structures in Denmark, their model differs from ours as in \cite{kmt19}, children of different ages compete for available slots. Similarly, \cite{kk18} explores a dynamic matching model where agents interact over time, focusing on agents with preferences on both sides of the market. In contrast, in our model, schools are endowed with priorities. Furthermore, the properties introduced to guarantee the stability of matchings differ from ours. Other related studies include \cite{d22}, which defines dynamic stability in environments where not all agents can be matched simultaneously and matchings are irreversible; \cite{bc13}, which examines the long-run properties of assignment rules in a dynamic matching problem; \cite{al07}, which explores a dynamic house allocation problem with a common set of agents across all periods; \cite{k14}, which investigates the centralized housing allocation problem with overlapping generations of agents; \cite{p13}, which examines teacher allocation to public schools with seniority-based priorities and demonstrates strategy-proofness of the DA mechanism in their context; \cite{fetal20}, which studies efficient slot reallocation after school admission cancellations; and \cite{u10}, which extends centralized matching for kidney exchanges to dynamically evolving agent pools.

Moreover, our paper contributes to the literature on the existence of stable matchings with complementarities and peer effects (e.g., \citealp{dm97}; \citealp{p12}; \citealp{py22}). The paper by \cite{setal22} closely aligns with ours in this line of research. They address the childcare allocation problem with siblings in Japan, resembling the matching with couples model. They propose an algorithm based on integer programming that does not theoretically guarantee stable matching, but experimental results consistently yield stable matchings.

Finally, while our results primarily address the allocation of children to childcare facilities, they also hold relevance for related problems involving agents with complementarities for bundled goods (e.g., \citealp{b11}; \citealp{npv16}; \citealp{nv23}). By studying these applications, we can further explore the implications of our approach. 

The paper is organized as follows: Section \ref{sec:Model} introduces the model, Section \ref{sec:results} presents our results, and Section \ref{sec:Concluding-remarks} concludes.

\section{Model\label{sec:Model}}

In our model, there is a finite set of students $I=\left\{ i_{1},i_{2},\text{\ensuremath{\dots}},i_{n}\right\}$
and a finite set of childcare facilities that we call schools $S=\left\{ s_{1},s_{2},\text{\ensuremath{\dots}},s_{m}\right\}$. 
The students interact over two periods. Each student $i$ can be assigned at most one object at each period $t\in\left\{ 1,2\right\} $. Each school $s\in S$ admits a maximum number of students each period $t$. We denote the non-negative integer \textit{capacity} at time $t$ by $q_{s}^{t}$. The capacity of school $s\in S$ over the two periods is denoted by $q_{s}=(q_{s}^{1},q_{s}^{2})$. Let $q=(q^{1},q^{2})$ be the vector of first and second-period quotas for all schools.

Each student $i\in I$ has a strict, transitive, and complete preference relation $P_{i}$ over the set of schools and the possibility of not attending any school, denoted by $\emptyset$. If $sP_{i}\emptyset$, then school $s$ is acceptable to student $i$, if $\emptyset P_{i}s$ then school $s$ is unacceptable to student $i$. Let $P$ be the profile of preferences over schools for all students.

Students are interested in the school next period. This temporal dimension of the problem is usually modeled by offering the student absolute priority to remain in the same school next period. In this paper, we remove the entanglement between the \textbf{school} allocation (period one) and the \textbf{priority} for the next period (period two). In period one, student $i$ acquires (possibly) a \textbf{school} and an absolute \textbf{priority} for the same school in the subsequent period. \footnote{Student $i$ may acquire only an absolute priority for the second period without securing a school in the first period.} It is important to note that the allocation of slots for periods 1 and 2 will occur through separate allocation processes at the beginning of each school year. 

We  denote as $\succ_{i}^{1}$ the preferences of student $i$ for schools in the first period and $\succ_{i}^{2}$ the preferences for priority in the second period. It can be represented by a strict ordering of the elements in \textbf{$\mathcal{S}:=\left[\left(S\cup\emptyset\right)\times\left(S\cup\emptyset\right)\right]$}, \textbf{school-priority pairs}, as $\left(\succ_{i}=(\succ_{i}^{1},\succ_{i}^{2})\right)_{i\in I}$, overall possible combination of ordered pairs of school-priority. We call \textbf{extended preferences} to the student preferences over school-priority pairs and denote it by $\mathcal{P}$. To simplify notation, we denote a generic element of $\mathcal{S}$ by $\left(\sigma_{i},\sigma_{j}\right)$ where $\sigma_{i},\sigma_{j}\in S\cup \{\emptyset\} $. Let $\succ$ denote all students' preferences profile over school-priority pairs.  

In the remainder of the paper, we restrict the possible student preferences over pairs of school-priority to the following type we observe in the childcare allocation problem in our hands. First, we consider the families who do not want to send their children to a school in the first year.
 
\begin{definition}
\label{def:resp_pref-1-1}
Extended preferences of a student $i\in I$ are \emph{priority-only} if $P_{i}=\succ_{i}^{2}$ and $\succ_{i}^{1}:\emptyset$.
\end{definition}

We consider the presence of \textbf{complementarities} between school (period one) and priority (period two), and assume that students prefer to attend a school only if they are allocated a slot for period one and have priority for period two in the same school. Otherwise, they prefer to remain at home.

\begin{definition}
\label{ass1:compl-1}
Extended preferences of a student $i\in I$ are \emph{willingness-to-remain} if for all
$s_{p},s_{r}\in S$:\\
$(i)$ $(s_{p},s_{p})\succ_{i}(s_{p},s_{r})$ for all $s_{p}\neq s_{r}$,
\\
$(ii)$ $(s_{p},s_{p})\succ_{i}(s_{r},s_{r})$ for all $s_{p}P_{i}s_{r}P_{i}\emptyset$,
and\\
$(iii)$ $(\emptyset,\emptyset)\succ_{i}(s_{p},s_{r})$, for all $s_{p}\neq s_{r}$.
\end{definition}

\begin{remark}
\label{rem1:notation}
With a slight abuse of notation, we write $\succ_{i}$ to denote $i$'s preferences over individual schools and allocations of school-priority pairs whenever there is no ambiguity.
\end{remark}

The union of the extended preferences previously defined is the domain of preferences where we state our results. We formally define the domain of reasonably extended preferences in Definition \ref{domain}.

\begin{definition}
\label{domain}
Let $\mathbf{P}$ denote the domain of reasonably extended preferences. For any $\succ_{i}\in\mathcal{P}$ and $i\in I$, the preference $\succ_{i}$ satisfies priority-only or willingness-to-remain.
\end{definition}

Each school $s\in S$ has a priority ordering $\pi_{s}$ over students. The priority ordering of a school $s$, $\pi_{s}=(\pi_{s}^{1},\pi_{s}^{2})$, does not change between periods, i.e., it is the same to allocate, schools and priorities, $\pi_{s}^{1}=\pi_{s}^{2}$. Let $\pi$ denote the profile of priorities of schools over students.  We assume that the priorities of each school over sets of students are responsive to the priorities of individual students. Let $\pi_{s}$ be the priority ordering of school $s$ over students. We say that $\pi_{s}$ is a \emph{responsive priority ordering}for all $I'\subseteq I$ with $|I'|<q_{s}$, and $i,i'$ in $I\setminus I'$ if it satisfies \emph{(i}) $I'\cup\{i\}\pi_{s}I'\cup\{i'\}$ if and only if $i\pi_{s}i'$ and \emph{(ii)} $I'\cup\{i\}\pi_{s}I'$ if and only if $i$ is acceptable to school $s$. Given a period $t\in\{1,2\}$, the choice of a school $s\in S$, $Ch_{s}^{t}:2^{I}\rightarrow2^{I}$, is induced by its priority ordering $\pi_{s}^{t}$ and quota $q_{s}^{t}$, \emph{i.e.,} $i\in Ch_{s}^{t}(I)$ if and only if there exists no set of students $I'\subseteq I\setminus\{i\}$ such that $|I'|=q_{s}^{t}$ and $i'\pi_{s}^{t}i$ for $i'\in I'$. The tuple $(I, S,\succ,\pi,q)$ describes a childcare allocation problem.

A \emph{matching} $\mu=(\mu^{1},\mu^{2})$ is a mapping defined on the set $I\cup S$ such that $(\mu^{1}(i),\mu^{2}(i))=\mu(i)\in S\cup\{\emptyset\}\times S\cup\{\emptyset\}$ for every $i\in I$, $(\mu^{1}(s),\mu^{2}(s))=\mu(s)\in2^{I}\times2^{I}$ for every $s\in S$, and satisfies
\begin{enumerate}
\item[(i)] $i\in\mu^{t}(s)$ if and only if $s=\mu^{t}(i)$ for $t=1,2$,
\item[(ii)]  $\mu(i)=(s,s')$ if and only $i\in\mu^{1}(s)$ and $i\in\mu^{2}(s')$,
\item[(iii)] $\mu^{t}(i)=\emptyset$ means  student $i$ is unassigned under $\mu$
at the period $t$ and $\mu^{t}(s)=\emptyset$ means that school $s$
is unassigned under $\mu$ at the period $t$,
\item[(iv)] $\mu(i)=(\emptyset,s)$ if and only if $i\in\mu^{2}(s)$ and $\mu^{1}(i)=\emptyset$;
and $\mu(i)=(s,\emptyset)$ if and only if $i\in\mu^{1}(s)$ and $\mu^{2}(i)=\emptyset$.
\end{enumerate}
$\mathcal{M}$ denotes the set of  all matchings. A matching is \emph{individually rational} if for no student $i\in I$, $\emptyset\succ_{i}\mu^{t}(i)$ for any $t\in\{1,2\}$, and for all schools $s\in S$, $Ch_{s}^{t}(\mu^{t}(s))=\mu^{t}(s)$ for $t=1,2$. For a given matching $\mu=(\mu^{1},\mu^{2})$, blocking coalitions can be formed in different ways:
\begin{itemize}
\item $(i,(s,\mu^{2}(i)))\in I\times S\cup\{\emptyset\}$ blocks $\mu$
if $s\succ_{i}\mu^{1}(i)$, $i\in Ch_{s}^{1}(\mu^{1}(s)\cup\{i\})$,
\item $(i,(\mu^{1}(i),s))\in I\times S\cup\{\emptyset\}$ blocks $\mu$
if $s\succ_{i}\mu^{2}(i)$, $i\in Ch_{s}^{2}(\mu^{2}(s)\cup\{i\})$,
\item $(i,(s,s'))\in I\times(S\cup\{\emptyset\}\times S\cup\{\emptyset\})$ blocks $\mu$ if $s\succ_{i}\mu^{1}(i),s'\succ_{i}\mu^{2}(i)$, and $i\in Ch_{s}^{1}(\mu^{1}(s)\cup\{i\})$, $i\in Ch_{s'}^{2}(\mu^{2}(s')\cup\{i\})$ with the possibility of $s=s'$.
\end{itemize}

Given a matching $\mu$, if a student's assignment at a period does not change under another matching $\mu'$, we abuse the notation and represent the unchanged assignment under $\mu'$ by ``--''. For instance, given the assignment of student $i$ under matchings $\mu$ and $\mu'$ be $\mu(i)=(s,s')$ and $\mu'(i)=(s,s'')$. Since student $i$ is assigned to school $s$ at period one both at $\mu$ and $\mu'$, when no confusion arises, we will write $\mu'(i)=(-,s'')$.

A matching $\mu$ is \emph{stable} if any coalition does not block it and is individually rational. A mechanism $\phi$ is a function that maps preference profiles to matchings. The matching obtained by mechanism $\phi$ at the preference profile $\succ$ is denoted by $\phi(\succ)$, where $\phi_{l}(\succ)$ represents the assignment of agent $l\in I\cup S$. We say that a mechanism is \emph{strategy-proof} if there does not exist a preference profile $\succ$ and an agent $l\in I\cup S$, and a preference profile $\succ'$ of agent $l$ such that $\phi_{l}(\succ'_{l},\succ_{-l})\succ_{l}\phi_{l}(\succ)$.

\section{Results}
\label{sec:results}
We have extended students' preferences over schools to allow them to express the possibility of joint school next period without having to attend school in the first period. In our domain $\mathbf{P}$ of reasonably extended preferences we can establish an analogy with the matching with couples model and treat students with priority-only extended preferences as ``single'' applicants. In contrast, a student with willingness-to-remain extended preferences takes the role of a ``couple''. The respective employers will be the schools in the first or the second period. Notice that a crucial difference with the matching with couples problem is that the first and second periods are different objects; a student can not abandon a slot in the second period for a slot in the first period. Therefore, no substitution is possible among slots in different periods. Next, we introduce the adapted $SPDA$ for our domain.
\smallskip\\
\noindent \textit{Adapted SPDA (aSPDA)}.\\
\noindent \textsc{Step 1.} Run $SPDA$ algorithm for the sub-market consisting of students with extended preferences satisfying priority-only, schools take into account only such students in their priority lists and only quotas for the second period $q_2$. Let $M$ be the set of student-school pairs tentatively matched to each other. Let $\mu$ be a matching for the initial problems such that the pairs in $M$ are matched to each other and unmatched otherwise.\\



\noindent \textsc{Step 2.} Fix a random order over the students whose extended preferences satisfy willingness-to-remain. Following the fixed order and given matching $\mu$ defined at the end of Step 1, introduce students individually to the initial market by running $SPDA$.   Each student $i$ applies to her remaining top choice until either a school accepts her or all schools reject her. If another student $i'$ is evicted from her school, then assign student $i$ tentatively to this school and $i'$ applies to her remaining top choice until either a school accepts her or all schools reject her. Update the matching $\mu$ after introducing each student in the order by tentatively assigning students to schools that accept them or students to become unmatched following $SPDA$. 
\smallskip\\
$aSPDA$ runs until no rejected students want to apply to further schools. 

Note that in \textsc{Step 2} quotas for the first and the second periods should be respected. Hence, a student whose extended preference satisfies priority-only still can evict a student whose extended preference satisfies willingness-to-remain if she has a higher priority. Moreover, due to the complementarities we observe, the students with willingness-to-remain are willing to not participate in the first period unless they can participate in both periods. Hence, no student can fill a seat at a school without hurting other students, and the obtained matching is not wasteful.

Example \ref{ex:childcare} illustrates how $aSPDA$ is executed.

\begin{example}
\label{ex:childcare}
Consider a childcare allocation problem $(I,S,\mathcal{P},\pi,q)$ with six students $I=\{i_{1},i_{2},i_{3},i_{4},i_{5},i_{6}\}$, two schools $S=\{s_{1},s_{2}\}$ and both schools have a capacity of one student in both periods, $q^{1}(s_{1},s_{2})=q^{2}(s_{1},s_{2})=(2,2)$. Their extended preferences are as follows:

$i_{1}:(s_{1},s_{1}),(s_{2},s_{2})$;

$i_{2}:(s_{1},s_{1}),(s_{2},s_{2})$;

$i_{3}:(s_{2},s_{2}),(s_{1},s_{1})$;

$i_{4}:(s_2,s_{2}),(s_1,s_{1})$;

$i_{5}:(\emptyset,s_{2}),(\emptyset,s_{1})$;

$i_{6}:(\emptyset,s_{1}),(\emptyset,s_{2})$.

\noindent The priorities of the schools are as follows:

$s_{1}:i_{3},i_{4},i_1,i_{2},i_{6},i_{5}$;

$s_{2}:i_{6},i_{2},i_{5},i_{3},i_{4},i_{1}$.

Notice that there are four students ($i_{1}$, $i_{2}$, $i_{3}$, $i_{4}$) whose extended preferences satisfy willingness-to-remain. Then, the $SPDA$ algorithm is run for two students whose extended preferences are priority-only, and schools consider these students in their priority order while respecting their quotas only in the second period. It assigns $i_6$ to school $s_{1}$ and $i_5$ to school $s_2$ in the second period. Then, the tentative assignment for this sub-market is $\{(i_{5},(\emptyset,s_{2})),(i_{6},(\emptyset,s_{1}))\}$. 

Next, fix, for instance, the order $\rho=(i_{1}, i_{4},i_{3},i_{2})$ over students whose preference satisfies willingness-to-remain. First, student $i_1$ applies to her best choice $(s_1,s_1)$ and tentatively assigned. Second, $i_4$ enters and applies to her best choice $(s_2,s_2)$ and tentatively assigned. Third, $i_3$ enters to the market and starts a rejection chain: $i_3$ evicts $i_4$ from $s_2$ as $i_3 \pi_{s_{2}} i_4$. Then, $i_4$ applies to her remaining best choice $s_1$ and evicts $i_6$. As a result, $i_6$ evicts $i_3$ from $s_2$ and $i_1$ from $s_1$. Since $i_1$ has the lowest priority among $i_5,i_6$ at school $s_2$, $i_1$ becomes unmatched. Finally, $i_2$ enters the market and, after being rejected by $s_1$, applies to $s_2$. As a consequence, $i_5$ is evicted from $s_2$. Then, she is rejected by $s_1$ and becomes unmatched. Hence, the final matching is $\mu=\{(i_{1},(\emptyset,\emptyset)),(i_{2},(s_2,s_2)),(i_{3},(s_1,s_1)),(i_{4},(s_1,s_1)),(i_{5},(\emptyset,\emptyset)),(i_{6},(\emptyset,s_2)\}$.
\end{example}

\begin{table}[!h] 
\begin{center}
\begin{tabular}{|c|c|c|c|c|c|c|}
\hline 
$s_{1}^{1}$ & $s_{1}^{2}$ & $s_{2}^{1}$ & $s_{2}^{2}$ & $\emptyset$  & $\rho$ & entrant \tabularnewline
\hline 
\hline 
 & $i_{6}$ &  & $i_{5}$ &   &  & $i_5,i_6$ \tabularnewline
\hline 
$i_{1}$ & $i_{1},i_{6}$ &  & $i_5$ &   & $(i_1,i_4,i_3,i_2)$ & $i_{1}$ \tabularnewline
\hline 
$i_{1}$ & $i_{1},i_{6}$ & $i_4$ & $i_4$, $i_5$ &   & $(i_1,i_4,i_3,i_2)$ & $i_{4}$ \tabularnewline
\hline 
$i_{4}$ & $i_{4}$ & $i_3$ & $i_3$, $i_5$ &   & $(i_1,i_4,i_3,i_2)$ & $i_{3}$ \tabularnewline
\hline 
$i_{3},i_{4}$ & $i_{3},i_{4}$ & $i_2$ & $i_2$, $i_6$ & $i_1,i_5$  & $(i_1,i_4,i_3,i_2)$ & $i_{2}$ \tabularnewline
\hline 
\end{tabular}
\end{center}
\caption{Execution of the mechanism in Example \ref{ex:childcare}.}
\label{Table1}
\end{table}

Table \ref{Table1} exhibits the execution of the mechanism for the childcare problem considered in Example \ref{ex:childcare}. Each row in Table \ref{Table1} represents a tentative assignment during the execution of the mechanism.

Note that in this example, considering each period (the possibility of a slot in period one and the priority for a slot in period two) as a separate market does not lead to a matching for the initial problem. If we take separate markets for period one, $SPDA$ will assign $i_1,i_2$ to $s_1$ in both periods and $i_3, i_4$ to $s_2$ in the first period, while $i_3, i_5$ to $s_2$ in the second period. Since $i_5$ and $i_3$ have higher priority at $s_2$ than $i_4$, the matching $\mu=\{(i_1,(s_1,s_1)),(i_2,(s_1,s_1)),(i_3,(s_2,s_2)),(i_4,(\emptyset,\emptyset)),(i_5,(\emptyset,s_2)),(i_6,(\emptyset,\emptyset))\}$ can be considered as a candidate. Nonetheless, it is not a stable matching for the initial problem. Hence, if we apply $SPDA$ at each period separately while respecting the priorities and the quotas, constructing a matching for the initial problem from the matchings as above may lead to an unstable matching. Meanwhile, $aSPDA$ does not exhibit this issue.

First, we show that in the domain of reasonably extended preferences $\mathbf{P}$, there exists a stable matching.

\begin{theorem}
\label{thm:stable}
Given a childcare allocation problem $(I,S,\mathcal{P},\pi,q)$ in the domain of reasonably extended preferences $\mathbf{P}$, there exists a stable matching.
\end{theorem}

\begin{proof}
Take the matching $\mu$, which results from $aSPDA$. The matching $\mu$ is individually rational since no student proposes to an unacceptable school. Next, we show no pair of $i\in I$ and $s\in S$ blocks $\mu$. Suppose, on the contrary, there exists such a pair of school $s$ and student $i$.

First, suppose that $i$ never proposes to school $s$. Then student $i$ cannot be unmatched if school $s$ is acceptable for her.  Otherwise, the algorithm would not end. Under the assumption that $i$ and $s$ blocks matching $\mu$, $s$ must be acceptable. Hence, $i$ has been matched to another school $s'$ during the $aSPDA$ such that $s'\mathcal{P}_{i}s$. It contradicts the assumption that $i$ and $s$ blocks $\mu$. 

Second, suppose that $i$ proposed to school $s$ during $aSPDA$. Then, $i$ must have been rejected by $s$. Thus, all students tentatively assigned to school $s$ at this round have higher priority than student $i$. Hence, $i$ and $s$ cannot block the matching $\mu$. 

In the meantime, a rejection chain might be started since we introduce students with willingness-to-remain in a fixed order. We need to show that there is no rejection cycle. In this case, students who have been rejected apply to their remaining best choice, which may lead to new rejections (including of tentatively matched students at a school). Following $aSPDA$, agents in the rejection chains apply to their best-remaining school until no rejected agents want to apply for a seat at a school. Since schools have capacity constraints and students apply to their remaining best choice in the order of their preferences, any rejection chain is finite. Hence, there is no rejection cycle.
\end{proof}

Next, we show that, in the domain $\mathbf{P}$, it is a dominant strategy for all students to submit their preferences truthfully.

\begin{theorem}
\label{thm:sp}
Given a childcare allocation problem $(I,S,\mathcal{P},\pi,q)$ in the domain of reasonably extended preferences $\mathbf{P}$, the $aSPDA$ mechanism is strategy-proof for students.
\end{theorem}

\begin{proof}
First, notice that a student cannot be better off by misreporting willingness-to-remain instead of priority-only, and vice versa. Next, suppose otherwise; i.e. there is an extended preference profile $\mathcal{P}'$ such that $\mathcal{P}'_{i}\neq \mathcal{P}_{i}$ for a student $i\in I$ and $\mathcal{P}'_{j}=\mathcal{P}_{j}$ for all other students $j\in I\setminus \{i\}$, the matching $\mu'$ is the matching obtained from $aSPDA$ from the profile $\mathcal{P}'$ and $\mu$ is the matching obtained from $aSPDA$ from the profile $\mathcal{P}$. Then, $\mu'(i)\succ_{i}\mu(i)$. Let $\hat{P_{i}}=(s_1,\ldots, s_{i-1})$ be the sequence of schools that student $i$ applies under her true preferences until she is accepted by the school $\mu(i)$. Then, student $i$ is rejected by $s_{i-1}$ before accepted by $\mu(i)$. Let also $\hat{\mathcal{P}'_{i}}$ be a sequence of schools such that student $i$ applies when she misrepresents her true preferences. For $\hat{\mathcal{P}'_{i}}$, we consider two possible cases.

First, suppose that $\mu'(i)$ is the least preferred school in $\hat{\mathcal{P}'_{i}}$ based on her true preferences. By construction of $\mathcal{P}'_{i}$, that is the untruthful preference profile of student $i$, school $\mu'(i)$ is ranked higher than $\mu(i)$. Then, all schools in the sequence $\hat{\mathcal{P}'_{i}}$ also appears in the sequence $\hat{\mathcal{P}_{i}}$ under her true preferences $\mathcal{P}_{i}$. Thus, the Scenario Lemma of \cite{df81} applies and leads to a contradiction.

Second, suppose that there exists a school $s\in S$ under her true preferences $\mathcal{P}_{i}$ such that she prefers less than $\mu'(i)$, i.e. $\mu'(i)\succ_{i} s$. Then, we can construct another sequence $\hat{\mathcal{P}''_{i}}$ of schools by removing from $\hat{\mathcal{P}'_{i}}$ all the schools that are less desired than $\mu'(i)$ under her true preferences. Then, if student $i$ is assigned to school $\mu'(i)$ following $\hat{\mathcal{P}''_{i}}$, as it is a smaller sequence than $\hat{\mathcal{P}'_{i}}$, the Scenario Lemma of \cite{df81} applies as in the first case and leads to a contradiction.
\end{proof}

\subsection{Maximal domain}
This subsection presents a maximal domain of preferences with the
relevant property. In a childcare allocation problem,
if at least one student's  preferences do not satisfy willingness-to-remain, we can construct preferences for other students satisfying this property such that no stable matching exists.

\begin{proposition}\label{pro:max_dom} 
The domain of reasonably extended preferences $\mathbf{P}$ is a maximal domain for the existence of stable matchings. 
\end{proposition}

\begin{proof}
We prove the proposition by constructing a counterexample
where dropping willingness-to-remain property does not have a stable matching.

Consider a childcare allocation problem with three students $I=\{i_{1},i_{2},i_{3}\}$, two schools $S=\{s_{1},s_{2}\}$ and the capacities of schools are $q^{1}(s_{1},s_{2})=q^{2}(s_{1},s_{2})=(1,2)$. The preferences of the students are as follows:

$i_{1}:(s_{1},s_{2}),(s_{2},s_{2})$,

$i_{2}:(s_{1},s_{1}),(s_{2},s_{2})$,

$i_{3}:(s_{2},s_{2})$.

The priorities of the schools are as follows:

$s_{1}:i_{1},i_{2}$,

$s_{2}:i_{2},i_{3},i_{1}$.

Notice that the (extended) preferences of student $i_{1}$ fail to
satisfy willingness-to-remain property since $(s_{1},s_{2})\succ_{i_{1}}(s_{2},s_{2})$.

We see that for each individually rational matching, there exists some
blocking coalitions:
\begin{itemize}
\item $\mu=\Big\{\Big(i_{1},(s_{1},s_{2})\Big),\Big(i_{2},(s_{2},s_{2})\Big)\Big\}$
is blocked by $\Big(i_{3},(s_{2},s_{2})\Big)$ with $s_{2}\succ_{i_{3}}^{1}\emptyset$,
$s_{2}\succ_{i_{3}}^{2}\emptyset$, and $i_{3}\pi_{s_{2}}^{1}\emptyset$,
$i_{3}\pi_{s_{2}}^{2}i_{1}$;
\item $\mu=\Big\{\Big(i_{1},(s_{1},s_{2})\Big),\Big(i_{3},(s_{2},s_{2})\Big)\Big\}$
is blocked by $\Big(i_{2},(s_{2},s_{2})\Big)$ with $s_{2}\succ_{i_{2}}^{1}\emptyset$,
$s_{2}\succ_{i_{2}}^{2}\emptyset$, and $i_{2}\pi_{s_{2}}^{1}\emptyset$,
$i_{2}\pi_{s_{2}}^{2}i_{1};$
\item $\mu=\Big\{\Big(i_{1},(s_{2},s_{2})\Big),\Big(i_{2},(s_{1},s_{1})\Big)\Big\}$
is blocked by $\Big(i_{3},(s_{2},s_{2})\Big)$ with $s_{2}\succ_{i_{3}}^{1}\emptyset$,
$s_{2}\succ_{i_{3}}^{2}\emptyset$, and $i_{3}\pi_{s_{2}}\emptyset$;
\item $\mu=\Big\{\Big(i_{1},(s_{2},s_{2})\Big),\Big(i_{2},(s_{1},s_{1})\Big),\Big(i_{3},(s_{2},s_{2})\Big)\Big\}$
is blocked by $\Big(i_{1},(s_{1},-)\Big)$ with $s_{1}\succ_{i_{1}}^{1}\emptyset$,
$s_{2}=_{i_{1}}^{2}s_{2}$ and $i_{1}\pi_{s_{1}}^{1}i_{2}$;
\item $\mu=\Big\{\Big(i_{1},(s_{2},s_{2})\Big),\Big(i_{2},(s_{2},s_{2})\Big)\Big\}$
is blocked by $\Big(i_{3},(s_{2},s_{2})\Big)$ with $s_{2}\succ_{i_{3}}^{1}\emptyset$,
$s_{2}\succ_{i_{3}}^{2}\emptyset$, and $i_{3}\pi_{s_{2}}i_{1}$;
\item $\mu=\Big\{\Big(i_{2},(s_{2},s_{2})\Big),\Big(i_{3},(s_{2},s_{2})\Big)\Big\}$
is blocked by $\Big(i_{2},(s_{1},s_{1})\Big)$ with $s_{1}\succ_{i_{2}}s_{2}$
and $i_{2}\pi_{s_{1}}\emptyset$.
\end{itemize}
Hence, there is no stable matching.\end{proof}

\section{Concluding remarks\label{sec:Concluding-remarks}}

We address the problem of assigning indivisible objects in the presence of complementarities among agents. Our focus is on the childcare allocation problem. Families have incentives to apply for childcare facility slots to secure priority for future allocations, even if they do not intend to utilize the facility in the initial period. This incentive creates an excess demand for childcare slots in the first period, leading to a distortion caused by these strategic considerations. 

A similar distortion manifests in the demand for feeder schools. Feeder schools are a well-known phenomenon in college admissions around the globe. \cite{aetal06} points out that most schools in Boston fill their slots according to a priority order such that the first group of students in the priority of a school consists of students attending a feeder school. Research by \cite{we07} examines the feeder legacy of high schools, demonstrating that students from feeder schools are more likely to attend specific colleges. Moreover, attending a feeder school can influence students' preferences, as evidenced by \cite{ntc06}, who find that graduates from feeder schools in Texas prefer selective institutions as their first choice compared to graduates from non-feeder schools.

To tackle the childcare allocation problem, following a minimalist market design approach (\citealp{s23}), we define a natural preference domain in which the $aSPDA$ mechanism consistently yields a stable matching. Additionally, we establish the strategy-proofness of the $aSPDA$ mechanism for students in this specific domain. Lastly, we present a counterexample highlighting how the preference domain satisfying the willingness-to-remain property represents a maximum domain in which stable matchings can exist.

\singlespacing
\bibliographystyle{te}
\bibliography{ARM_Childcare}
\end{document}